%% file: main.tex
\documentclass[11pt]{article}

\usepackage{a4}
\usepackage{amsthm}
\usepackage{amsmath,amssymb}
\usepackage{tabls}
\usepackage[dvipdfmx]{graphicx}
\usepackage{wrapfig}
\usepackage{array}
\usepackage[algoruled,linesnumbered,algo2e,vlined]{algorithm2e}
\usepackage{url}
\usepackage{multirow}
\usepackage{multicol}
\usepackage{cases}
\usepackage[normalem]{ulem}
\usepackage{authblk}

\newtheorem{theorem}{Theorem}
\newtheorem{lemma}{Lemma}

\theoremstyle{remark}
\newtheorem{definition}{Definition}

\input{macros}

\begin{document}

\title{
  The Parameterized Position Heap of a Trie\thanks{YN, SI, HB, MT are respectively supported by JSPS KAKENHI Grant Numbers JP18K18002, JP17H01697, JP16H02783, JP18H04098.}
}

\date{}

\author{Noriki~Fujisato}
\author{Yuto~Nakashima}
\author{Shunsuke~Inenaga}
\author{Hideo~Bannai}
\author{Masayuki~Takeda}

\affil{Department of Informatics, Kyushu University, Japan\\
  \texttt{\{noriki.fujisato, yuto.nakashima, inenaga, bannai, takeda\}@inf.kyushu-u.ac.jp}}

\maketitle

\begin{abstract}
  Let $\Sigma$ and $\Pi$ be disjoint alphabets of respective size $\sigma$ and $\pi$.
  Two strings over $\Sigma \cup \Pi$
  of equal length are said to \emph{parameterized match} (\emph{p-match})
	if there is a bijection $f:\Sigma \cup \Pi \rightarrow \Sigma \cup \Pi$
  such that (1) $f$ is identity on $\Sigma$ and
  (2) $f$ maps the characters of one string to those
	of the other string so that the two strings become identical.
  We consider the p-matching problem
  on a (reversed) trie $\mathcal{T}$ and a string pattern $P$
  such that every path that p-matches $P$ has to be reported.
  Let $N$ be the size of the given trie $\mathcal{T}$.
  In this paper, we propose the \emph{parameterized position heap} for $\mathcal{T}$
  that occupies $O(N)$ space and 
  supports p-matching queries in $O(m \log (\sigma + \pi) + m \pi + \pocc))$ time,
	where $m$ is the length of a query pattern $P$
	and $\pocc$ is the number of paths in $\mathcal{T}$ to report.
  We also present an algorithm which constructs
  the parameterized position heap for a given trie $\mathcal{T}$ 
  in $O(N (\sigma + \pi))$ time and working space.
\end{abstract}

\input{introduction}
\input{preliminaries}
\input{pph}

\input{construction}

\input{conclusions}

\bibliographystyle{abbrv}
\bibliography{ref}

\end{document}

%% file: macros.tex
\newcommand{\Prefix}{\mathsf{Prefix}}

\newcommand{\SHT}{\mathsf{SHT}}
\newcommand{\PPH}{\mathsf{PPH}}

\newcommand{\id}{\mathsf{id}}

\newcommand{\rslink}{\mathsf{rsl}}
\newcommand{\mrp}{\mathsf{mrp}}

\newcommand{\occ}{\mathit{occ}}
\newcommand{\pocc}{\mathit{pocc}}

\newcommand{\inputtrie}{\mathcal{T}}
\newcommand{\CST}{\mathsf{CST}}
\newcommand{\pCST}{\mathsf{pCST}(\mathcal{T})}
\newcommand{\pstrset}{\mathcal{W}_{\mathcal{T}}}
\newcommand{\spe}{\mathsf{spe}}
\newcommand{\pcs}{\mathsf{pcs}}

\newcommand{\cnode}[1]{\mathit{c}_{#1}}
\newcommand{\pnode}[1]{\mathit{h}_{#1}}

%% file: introduction.tex
\section{Introduction} \label{sec:intro}

The \emph{parameterized matching problem} (\emph{p-matching problem}),
first introduced by Baker~\cite{Baker96},
is a variant of pattern matching which looks for
substrings of a text that has ``the same structure'' as a given pattern.
More formally, we consider a parameterized string (p-string)
that can contain static characters from an alphabet $\Sigma$ and
parameter characters from another alphabet $\Pi$.
Two equal length p-strings $x$ and $y$ over the alphabet $\Sigma \cup \Pi$
are said to \emph{parameterized match} (\emph{p-match})
if $x$ can be transformed to $y$ (and vice versa)
by applying a bijection which renames the parameter characters.
The \emph{p-matching problem} is, given a text p-string $w$ and pattern p-string $p$,
to report the occurrences of substrings of $w$ that p-match $p$.
Studying the p-matching problem is well motivated by
plagiarism detection, software maintenance,
and RNA structural pattern matching~\cite{Baker96,Shibuya04}.
We refer readers to~\cite{MendivelsoP15} for detailed descriptions
about these motivations.

Baker~\cite{Baker96} proposed an indexing data structure for
the p-matching problem, called the \emph{parameterized suffix tree} (\emph{p-suffix tree}).
The p-suffix tree supports p-matching queries in $O(m \log(\sigma + \pi) + \pocc)$ time,
where $m$ is the length of pattern $p$,
$\sigma$ and $\pi$ are respectively the sizes of the alphabets $\Sigma$ and $\Pi$,
and $\pocc$ is the number of occurrences to report~\cite{Baker93}.
She also showed an algorithm that builds the p-suffix tree
for a given text $S$ of length $n$ in $O(n(\pi + \log \sigma))$ time
with $O(n)$ space~\cite{Baker96}.
Later, Kosaraju~\cite{Kosaraju95} proposed an algorithm to build the p-suffix tree
in $O(n \log(\sigma + \pi))$ time\footnote{The original claimed time bounds in Kosaraju~\cite{Kosaraju95} and in Shibuya~\cite{Shibuya04} are $O(n (\log \sigma + \log \pi))$. However, assuming by symmetry that $\sigma \geq \pi$, we have $\log \sigma + \log \pi = \log (\sigma \pi) \leq \log\sigma^2 = 2 \log \sigma = O(\log \sigma)$ and $\log (\sigma + \pi) \leq \log (2\sigma) = \log2 + \log \sigma = O(\log \sigma)$.} with $O(n)$ space.
Their algorithms are both based on McCreight's suffix tree construction algorithm~\cite{McC76},
and hence are \emph{offline} (namely, the whole text has to be known beforehand).
Shibuya~\cite{Shibuya04} gave an \emph{left-to-right online} algorithm that builds
the p-suffix tree in $O(n \log(\sigma + \pi))$ time with $O(n)$ space.
His algorithm is based on Ukkonen's suffix tree construction algorithm~\cite{Ukk95}
which scans the input text from left to right.

Diptarama et al.~\cite{DiptaramaKONS17} proposed a new indexing structure
called the \emph{parameterized position heap} (\emph{p-position heap}).
They showed how to construct the p-position heap of a given p-string $S$ of length $n$
in $O(n \log(\sigma + \pi))$ time with $O(n)$ space
in a \emph{left-to-right online} manner.
Their algorithm is based on Kucherov's position heap construction algorithm~\cite{Kucherov13}
which scans the input text from left to right.
Recently, Fujisato et al.~\cite{FujisatoNIBT18} presented another variant
of the p-position heap that can be constructed in a \emph{right-to-left online} manner,
in $O(n \log(\sigma + \pi))$ time with $O(n)$ space.
This algorithm is based on Ehrenfeucht et al.'s algorithm~\cite{ehrenfeucht_position_heaps_2011}
which scans the input text from right to left.
Both versions of p-positions heaps support
p-matching queries in $O(m \log(\sigma + \pi) + m\pi + \pocc)$ time.

This paper deals with indexing on \emph{multiple} texts;
in particular, we consider the case where those multiple texts are represented by a \emph{trie}.
It should be noted that our trie is a so-called \emph{common suffix trie} (\emph{CS trie})
where the common suffixes of the texts are merged
and the edges are reversed (namely, each text is represented by a path from a leaf to the root).
See also Figure~\ref{fig:cst} for an example of a CS trie.
There are two merits in representing multiple texts by a CS trie:
Let $N$ be the size of the CS trie of the multiple strings of total length $Z$.
(1)
$N$ can be as small as $\Theta(\sqrt{Z})$
when the multiple texts share a lot of common long suffixes.
(2) The number of distinct suffixes of the texts
is equal to the number of the nodes in the CS trie, namely $N$.
On the other hand, this is not the case with the ordinal common prefix trie (CP trie),
namely, the number of distinct suffixes in the CP trie
can be super-linear in the number of its nodes.
Since most, if not all, indexing structures require space
that is dependent of the number of distinct suffixes,
the CS trie is a more space economical representation for indexing than its CP trie counterpart.

Let $N$ be the size of a given CS trie.
Due to Property (1) above,
it is significant to construct an indexing structure \emph{directly} from the CS trie.
Note that if we expand all texts from the CS trie,
then the total string length can blow up to $O(N^2)$.
Breslauer~\cite{breslauer_suffix_tree_tree_1998} introduced the suffix tree for a CS trie
which occupies $O(N)$ space,
and proposed an algorithm which constructs it in $O(N\sigma)$ time and working space.
Using the suffix tree of a CS trie,
one can report all paths of the CS trie that exactly matches with a given pattern
of length $m$ in $O(m \log \sigma + \occ)$ time,
where $\occ$ is the number of such paths to report.
Shibuya~\cite{Shibuya_construct_stree_of_tree} gave an optimal $O(N)$-time construction
for the suffix tree for a CS trie in the case of integer alphabets of size $N^{O(1)}$.
Nakashima et al.~\cite{position_heaps_of_trie_2012} proposed
the position heap for a CS trie,
which can be built in $O(N\sigma)$ time and working space
and supports exact pattern matching in $O(m \log \sigma + \occ)$ time.
Later, an optimal $O(N)$-time construction algorithm for
the position heap for a CS trie
in the case of integer alphabets of size $N^{O(1)}$
was presented~\cite{NakashimaIIBT15}.

In this paper, we propose the \emph{parameterized position heap} for a CS trie $\mathcal{T}$,
denoted by $\PPH(\mathcal{T})$,
which is the \emph{first} indexing structure for p-matching on a trie.
We show that $\PPH(\mathcal{T})$ occupies $O(N)$ space,
supports p-matching queries in $O(m \log (\sigma + \pi) + m\pi + \pocc)$ time,
and can be constructed in $O(N(\sigma + \pi))$ time and working space.
Hence, we achieve optimal pattern matching and construction in the case of
constant-sized alphabets.
The proposed construction algorithm is fairly simple,
yet uses a non-trivial idea that converts a given CS trie into a smaller trie
based on the p-matching equivalence.
The simplicity of our construction algorithm comes from the fact that
each string stored in (p-)position heaps is represented by an explicit node,
while it is not the case with (p-)suffix trees.
This nice property makes it easier and natural to adopt
the approaches by Brealauer~\cite{breslauer_suffix_tree_tree_1998}
and by Fujisato et al.~\cite{FujisatoNIBT18} that use
\emph{reversed suffix links} in order to process the texts from left to right.
We also remark that all existing p-suffix tree construction
algorithms~\cite{Baker96,Kosaraju95,Shibuya04}
in the case of a single text require somewhat involved data structures due to non-monotonicity
of parameterized suffix links~\cite{Baker93,Baker96},
but our p-position heap does not need such a data structure
even in the case of CS tries (this will also be discussed in the concluding section).

%% file: preliminaries.tex
\section{Preliminaries}

Let $\Sigma$ and $\Pi$ be disjoint ordered sets
called a \emph{static alphabet} and a \emph{parameterized alphabet},
respectively.
Let $\sigma = |\Sigma|$ and $\pi = |\Pi|$.
An element of $\Sigma$ is called an \emph{s-character},
and that of $\Pi$ is called a \emph{p-character}.
In the sequel, both an s-character and a p-character
are sometimes simply called a \emph{character}.
An element of $\Sigma^*$ is called a \emph{string},
and an element of $(\Sigma \cup \Pi)^*$ is called a \emph{p-string}.
The length of a (p-)string $w$ is the number of
characters contained in $w$.
The empty string $\varepsilon$ is a string of length 0,
namely, $|\varepsilon| = 0$.
For a (p-)string $w = xyz$, $x$, $y$ and $z$ are called
a \emph{prefix}, \emph{substring}, and \emph{suffix} of $w$, respectively.
The set of prefixes of a (p-)string $w$ is denoted by $\Prefix(w)$.
The $i$-th character of a (p-)string $w$ is denoted by
$w[i]$ for $1 \leq i \leq |w|$,
and the substring of a (p-)string $w$ that begins at position $i$ and
ends at position $j$ is denoted by $w[i..j]$ for $1 \leq i \leq j \leq |w|$.
For convenience, let $w[i..j] = \varepsilon$ if $j < i$.
Also, let $w[i..] = w[i..|w|]$ for any $1 \leq i \leq |w|$.
For any (p-)string $w$, let $w^R$ denote the reversed string of $w$, i.e., $w^R = w[|w|] \cdots w[1]$.

Two p-strings $x$ and $y$ of length $k$ each
are said to \emph{parameterized match} (\emph{p-match})
iff there is a bijection $f$ on $\Sigma \cup \Pi$
such that $f(a) = a$ for any $a \in \Sigma$
and $f(x[i]) = y[i]$ for all $1 \leq i \leq k$.
For instance,
let $\Sigma = \{\mathtt{a}, \mathtt{b}\}$ and $\Pi = \{\mathtt{x}, \mathtt{y}, \mathtt{z}\}$,
and consider two p-strings $x = \mathtt{axbzzayx}$ and $y = \mathtt{azbyyaxz}$.
These two strings p-match,
since $x$ can be transformed to $y$
by applying a renaming bijection $f$
such that $f(\mathtt{a}) = \mathtt{a}$, $f(\mathtt{b}) = \mathtt{b}$,
$f(\mathtt{x}) = \mathtt{z}$, $f(\mathtt{y}) = \mathtt{x}$, and $f(\mathtt{z}) = \mathtt{y}$
to the characters in $x$.
We write $x \approx y$ iff two p-strings $x$ and $y$ p-match.
It is clear that $\approx$ is an equivalence relation on p-strings over $\Sigma \cup \Pi$.
We denote by $[x]$ the equivalence class for p-string $x$ w.r.t. $\approx$.
The representative of $[x]$ is the lexicographically smallest p-string in $[x]$,
which is denoted by $\spe(x)$.
It is clear that two p-strings $x$ and $y$ p-match iff $\spe(x) = \spe(y)$.
In the running example, $\spe(\mathtt{axbzzayx}) = \spe(\mathtt{azbyyaxz}) = \mathtt{axbyyazx}$

A \emph{common suffix trie} (\emph{CS trie}) $\mathcal{T}$ is a reversed trie
such that (1) each edge is directed towards the root,
(2) each edge is labeled with a character from $\Sigma \cup \Pi$,
and (3) the labels of the in-coming edges to each node are mutually distinct.
Each node of the trie represents the (p-)string obtained by concatenating 
the labels on the path from the node to the root.
An example of a CS trie is illustrated in Figure~\ref{fig:cst}.
$\CST(W)$ denotes the CS trie which represents a set $W$ of (p-)strings.

\begin{figure}[t]
  \centerline{
    \includegraphics[width=0.6\textwidth]{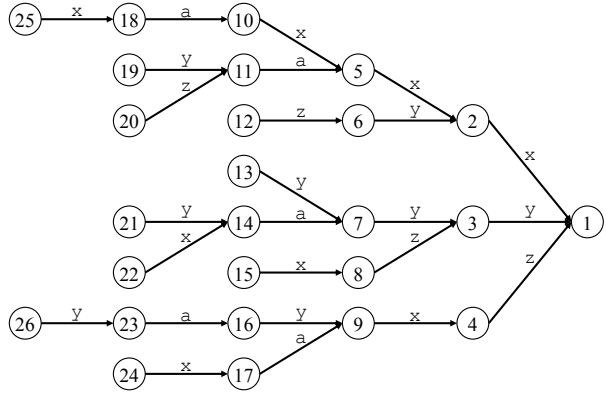}
  }
  \caption{
    CS trie for a set 
    $\{\mathtt{xaxxx}, \mathtt{yaxx}, \mathtt{zaxx}, \mathtt{zyx}, \mathtt{yyy}, 
    \mathtt{yayy}, \mathtt{xayy}, \mathtt{xzy}, \mathtt{yayxz}$, $\mathtt{xaxz}\}$
    of 10 p-strings over $\Sigma \cup \Pi$, where
    $\Sigma = \{\mathtt{a}\}$ and $\Pi = \{\mathtt{x}, \mathtt{y}, \mathtt{z}\}$.
  }
  \label{fig:cst}
\end{figure}

%% file: pph.tex
\section{Parameterized position heap of a common suffix trie}

In this section, we introduce the parameterized pattern matching (p-matching) problem  
on a common suffix trie that represents a set of p-strings, 
and propose an indexing data structure called a parameterized position heap of a trie.

\subsection{p-matching problem on a common suffix trie}
We introduce the \emph{p-matching problem} on a common suffix trie $\inputtrie$ and a pattern $p$.
We will say that a node $v$ in a common suffix trie p-matches with a pattern p-string $p$
if the prefix of length $|p|$ of the p-string represented by $v$ and $p$ p-match.
In this problem,
we preprocess a given common suffix trie $\inputtrie$ so that later,
given a query pattern $p$, 
we can quickly answer every node $v$ of $\inputtrie$
whose prefix of length $|p|$ and $p$ p-match.
For the common suffix trie in Figure~\ref{fig:cst},
when given query pattern $P = \mathtt{azy}$,
then we answer the nodes $17$ and $23$.

Let $W_{\inputtrie}$ be the set of all p-strings represented by nodes of $\inputtrie$.
By the definition of the common suffix trie, 
there may exist two or more nodes which represent different p-strings, but p-match.
We consider the common suffix trie which merges such nodes into the same node 
by using the representative of the parameterized equivalent class of these strings.
We define the set $\pcs(\inputtrie)$ of p-strings as follows:
$\pcs(\inputtrie) = \{\spe(w^R)^R \mid w \in W_{\inputtrie} \}$.
Then, the reversed trie which we want to consider is $\CST(\pcs(\inputtrie))$.
We refer to this reversed trie as the \emph{parameterized-common suffix trie} of $\inputtrie$, 
and denote it by $\pCST$ (i.e., $\pCST = \CST(\pcs(\inputtrie))$).
Each node of $\pCST$ stores pointers to its corresponding node(s) of $\inputtrie$. 
Then, by solving the p-matching problem on $\pCST$,
we can immediately answering p-matching queries on $\inputtrie$.
Figure~\ref{fig:p-cst} shows an example of $\pCST$.
In the rest of this paper, $N$ denotes the number of nodes of $\inputtrie$ 
and $N_p$ denotes the number of nodes of $\pCST$.
Note that $N \geq N_p$ always holds.

\begin{figure}[t]
  \centerline{
    \includegraphics[width=0.6\textwidth]{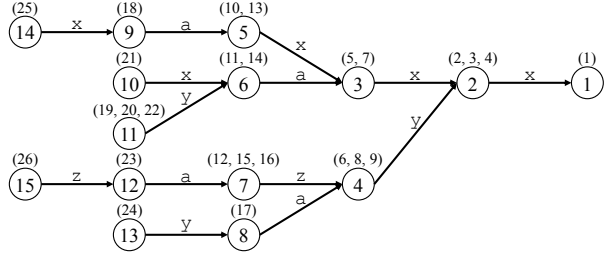}
  }
  \caption{
    Illustration of $\pCST$ for $\inputtrie$ 
    (where $\inputtrie$ is the common suffix trie illustrated in Figure~\ref{fig:cst}).
    Each node of $\pCST$ corresponds to nodes of $\inputtrie$ 
    which are labeled by elements in the tuple above the node of $\pCST$.
    For example, the node of $\pCST$ labeled 6 corresponds to 
    the nodes of $\inputtrie$ labeled 11 and 14.
  }
  \label{fig:p-cst}
\end{figure}

\subsection{Parameterized position heap of a common suffix trie}

Let $\mathcal{S} = \langle s_1, \ldots, s_k \rangle$ be a sequence of strings
such that for any $1 < i \leq k$, 
$s_i \not \in \Prefix(s_j)$ for any $1 \leq j < i$.

\begin{definition}[Sequence hash trees~\cite{coffman}]
\label{def:sev_hash}
The \emph{sequence hash tree} of a sequence
$\mathcal{S} = \langle s_1, \ldots, s_k \rangle$ of strings,
denoted $\SHT(\mathcal{S}) = \SHT(\mathcal{S})^{k}$, is a trie structure that is recursively defined as follows:
Let $\SHT(\mathcal{S})^{i} = (V_i, E_i)$.
Then 
\[
 \SHT(\mathcal{S})^{i} = 
  \begin{cases}
   (\{\varepsilon\}, \emptyset) & \mbox{if $i = 1$}, \\
   (V_{i-1} \cup \{u_{i}\}, E_{i-1} \cup \{(v_{i}, a, u_{i})\}) & \mbox{if $2 \leq i \leq k$},
  \end{cases}
\]
where $v_{i}$ is the longest prefix of $s_i$ which satisfies $v_{i} \in V_{i-1}$,
$a = s_i[|v_{i}|+1]$,
and $u_{i}$ is the shortest prefix of $s_i$ which satisfies $u_{i} \notin V_{i-1}$.
\end{definition}
Note that since we have assumed that each $s_i \in \mathcal{S}$ is not a prefix of 
$s_j$ for any $1 \leq j < i$, 
the new node $u_i$ and new edge $(v_{i}, a, u_{i})$ always exist 
for each $1 \leq i \leq k$.
Clearly $\SHT(\mathcal{S})$ contains $k$ nodes (including the root).

Let $\pstrset = \langle \spe(w_1), \ldots, \spe(w_{N_p}) \rangle$ be a sequence of p-strings 
such that $\{w_1, \ldots, w_{N_p}\} = \pcs(\inputtrie)$ and $|w_i| \leq |w_{i+1}|$ for any $1 \leq i \leq N_p - 1$.
$\pstrset(i)$ denote the sequence $\langle \spe(w_1), \ldots, \spe(w_{i}) \rangle$ for any $1 \leq i \leq N_p$, 
and $\pCST^i$ denote the common suffix trie of $\{\spe(w_1), \ldots, \spe(w_i)\}$, 
namely, $\pCST^i = \CST(\{\spe(w_1), \ldots, \spe(w_i)\})$.
The node of $\pCST$ which represents $w_i$ is denoted by $\cnode{i}$.
Then, our indexing data structure is defined as follows.

\begin{definition}[Parameterized positions heaps of a CST]\label{def:p_position_heap}
The \emph{parameterized position heap} (\emph{p-position heap})
for a common suffix trie $\inputtrie$, 
denoted by $\PPH(\inputtrie)$, is the sequence hash tree of $\pstrset$
i.e., $\PPH(\inputtrie) = \SHT(\pstrset)$.
\end{definition}

Let $\PPH(\inputtrie)^{i} = \SHT(\pstrset(i))$ for any $1 \leq i \leq N_p$
(i.e., $\PPH(\inputtrie)^{N_p} = \PPH(\inputtrie)$). 
The following lemma shows the exact size of $\PPH(\inputtrie)$.

\begin{lemma} \label{lem:position_nodes_correspondence}
  For any common suffix trie $\inputtrie$ such that the size of $\pCST$ is $N_p$,
  $\PPH(\inputtrie)$ consists of exactly $N_p$ nodes.
  Also, there is a one-to-one correspondence between
  the nodes of $\pCST$ and the nodes of $\PPH(\inputtrie)$.
\end{lemma}

\begin{proof}
  Initially, $\PPH(\inputtrie)^1$ consists only of the root
  that represents $\varepsilon$ since $w_1 = \varepsilon$.
  Let $i$ be an integer in $[1..N_p]$.
  Since $w_i$ does not p-match with $w_j$ 
  and $|\spe(w_i)| \geq |\spe(w_j)|$ for any $1 \leq j < i$, 
     there is a prefix of $\spe(w_i)$ 
  that is not represented by any node of $\PPH(\inputtrie)^{i-1}$.
  Therefore, when we construct $\PPH(\inputtrie)^i$ from $\PPH(\inputtrie)^{i-1}$,
  then exactly one node is inserted, which corresponds to the node representing $w_i$.
\end{proof}

Let $\pnode{i}$ be the node of $\PPH(\inputtrie)$ which corresponds to $w_i$.
For any p-string $p \in (\Sigma \cup \Pi)^+$,
we say that $p$ is \emph{represented} by $\PPH(\inputtrie)$
iff $\PPH(\inputtrie)$ has a path which starts from the root and spells out $p$.

Ehrenfeucht et al.~\cite{ehrenfeucht_position_heaps_2011}
introduced \emph{maximal reach pointers},
which are used for efficient pattern matching queries on position heaps.
Diptarama et al.~\cite{DiptaramaKONS17} and Fujisato et al.~\cite{FujisatoNIBT18}
also introduced maximal reach pointers for their p-position heaps,
and showed how efficient pattern matching queries can be done.
We can naturally extend the notion of maximal reach pointers
to our p-position heaps:

\begin{definition}[Maximal reach pointers]
  For each $1 \leq i \leq N_p$,
  the \emph{maximal reach pointer} of the node $\pnode{i}$
  points to the deepest node $v$ of $\PPH(\inputtrie)$ 
  such that $v$ represents a prefix of $\spe(w_i)$.
\end{definition}

The node which is pointed by the maximal reach pointer of node $\pnode{i}$ is denoted by $\mrp(i)$.
The \emph{augmented} $\PPH(\inputtrie)$ is $\PPH(\inputtrie)$
with the maximal reach pointers of all nodes.
For simplicity, if $\mrp(i)$ is equal to $\pnode{i}$,
then we omit this pointer.
See Figure~\ref{fig:aug-pph}
for an example of augmented $\PPH(\inputtrie)$.

\begin{figure}[tb]
  \centerline{
    \begin{tabular}{|l|r|} \hline
      $\spe(w_{1})$ & $\underline{\varepsilon}$ \\ \hline
      $\spe(w_{2})$ & $\underline{\mathtt{x}}$ \\ \hline
      $\spe(w_{3})$ & $\underline{\mathtt{xx}}$ \\ \hline
      $\spe(w_{4})$ & $\underline{\mathtt{xy}}$ \\ \hline
      $\spe(w_{5})$ & $\underline{\mathtt{xxx}}$ \\ \hline
      $\spe(w_{6})$ & $\underline{\mathtt{axx}}$ \\ \hline
      $\spe(w_{7})$ & $\underline{\mathtt{xyz}}$ \\ \hline
      $\spe(w_{8})$ & $\underline{\mathtt{axy}}$ \\ \hline
      $\spe(w_{9})$ & $\underline{\mathtt{axx}}\mathtt{x}$ \\ \hline
      $\spe(w_{10})$ & $\underline{\mathtt{xaxx}}$ \\ \hline
      $\spe(w_{11})$ & $\underline{\mathtt{xay}}\mathtt{y}$ \\ \hline
      $\spe(w_{12})$ & $\underline{\mathtt{axy}}\mathtt{z}$ \\ \hline
      $\spe(w_{13})$ & $\underline{\mathtt{xaxy}}$ \\ \hline
      $\spe(w_{14})$ & $\underline{\mathtt{xaxx}}\mathtt{x}$ \\ \hline
      $\spe(w_{15})$ & $\underline{\mathtt{xaxy}}\mathtt{x}$ \\ \hline
    \end{tabular}
    \hfill
    \raisebox{-25mm}{\includegraphics[scale=0.5]{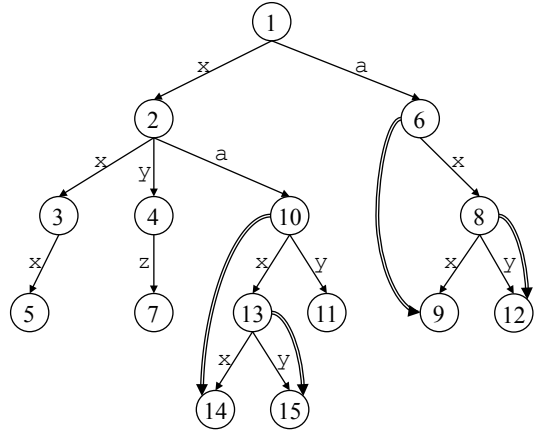}}
  }
  \caption{
    To the left is the list of $\spe(w_i)$
    for p-strings represented by $\pCST$ of Figure~\ref{fig:p-cst},
    where $\Sigma = \{\mathtt{a}\}$ and $\Pi = \{\mathtt{x}, \mathtt{y}, \mathtt{z}\}$.
    To the right is an illustration for augmented $\PPH(\inputtrie)$ 
    where the maximal reach pointers are indicated by the double-lined arrows.
    The underlined prefix of each $\spe(w_i)$ in the left list
    denotes the longest prefix of $\spe(w_i)$ that was represented in
    $\PPH(\inputtrie)$ and hence, the maximal reach pointer of the node with label $i$
    points to the node which represents this underlined prefix of $\spe(w_i)$.
  }
  \label{fig:aug-pph}
\end{figure}

\subsection{P-matching with augmented parameterized position heap}

It is straightforward that by applying Diptarama et al.'s pattern matching algorithm to
our $\PPH(\inputtrie)$ augmented with maximal reach pointers,
parameterized pattern matching can be done 
in $O(m \log (\sigma + \pi) + m \pi + \pocc')$ time 
where $\pocc'$ is the number of nodes in $\pCST$ that p-match with the pattern.
Since each node in $\pCST$ stores the pointers to the corresponding nodes in $\inputtrie$, 
then we can answer all the nodes that p-match with the pattern.

Diptarama et al.'s algorithm stands on Lemmas~{13} and 14 of \cite{DiptaramaKONS17}.
These lemmas can be extended to our $\PPH(\inputtrie)$ as follows:

\begin{lemma} \label{lem:pm_lemma_1}
  Suppose $\spe(p)$ is represented by a node $u$ of augmented $\PPH(\inputtrie)$.
  Then $p$ p-matches with the prefix of length $|p|$ of $w_i$
  iff $\mrp(i)$ is $u$ or a descendant of $u$.
\end{lemma}

\begin{proof}
   Let $u$ be the node in augmented $\PPH(\inputtrie)$ that represents $\spe(p)$.

   Assume that $p$ p-matches with the prefix of length $|p|$ of $w_i$ 
   and the node $v$ satisfying $\id(v) = i$ represents $\spe(w_i)[1..k]$.
   Then either $\spe(w_i)[1..k]$ is a prefix of $\spe(p)$ 
   or $\spe(p)$ is a prefix of $\spe(w_i)[1..k]$.
   This implies that $v$ is either an ancestor or a descendant of $u$.
   If $v$ is an ancestor of $u$, then $v$ ($\mrp(i)$) points to $u$ or an its descendant 
   since $\spe(w_i)[1..|p|]$ is represented by $u$.  
   By the definition of maximal reach pointers, 
   if $v$ is a descendant of $u$, $v$ ($\mrp(i)$) points to $u$ or an its descendant.

   Assume that $\mrp(i)$ is $u$ or an its descendant.
   Let $k$ be the integer such that $\mrp(i)$ represents $\spe(w_i)[1..k]$.
   Then $\spe(p)$ is a prefix of $\spe(w_i)[1..k]$.
   This implies that $p$ p-matches with the prefix of length $|p|$ of $w_i$.
\end{proof}

\begin{lemma} \label{lem:pm_lemma_2}
  Suppose that $\spe(p)$ is not represented in augmented $\PPH(\inputtrie)$.
  There is a factorization $q_1, \ldots, q_k$ of $p$ 
  s.t. $q_j$ is the longest prefix of $\spe(p[|q_1 \cdots q_{j-1}| + 1..|p|])$ 
  that is represented in augmented $\PPH(\inputtrie)$. 
  If $p$ p-matches with the prefix of length $|p|$ of $w_i$, 
  then $\mrp(i + |q_1 \cdots q_{j-1}|)$ is the node which represents $\spe(q_j)$ for any $1 \leq j < k$ 
  and $\mrp(i + |q_1 \cdots q_{k-1}|)$ is the node 
  which represents $\spe(q_k)$ or a descendant of $\mrp(i + |q_1 \cdots q_{k-1}|)$.
\end{lemma}

\begin{proof}
   Assume that $p = q_1, \ldots, q_k$ p-matches with the prefix of length $|p|$ of $w_i$.
   Since $\spe(q_1)$ is a prefix of $\spe(p)$,
   then $\mrp(i)$ is the node which represents $\spe(q_1)$ or an its descendant.
   If $\mrp(i)$ is an descendant of the node which represents $\spe(q_1)$, 
   then $q_1$ is not the longest prefix of $\spe(p)$ that is represented in augmented $\PPH(\inputtrie)$.
   Thus $\mrp(i)$ is the node which represents $\spe(q_1)$.
   Similarly, for every $1 < j < k$, $\spe(q_j)$ is a prefix of $\spe(p[|q_1 \cdots q_{j-1}| + 1..|p|])$ 
   and p-matches with the prefix of length $|q_j|$ of $w_i[|q_1 \cdots q_{j-1}|+1..|w_i|]$.
   Thus $\mrp(i + |q_1 \cdots q_{j-1}|)$ is the node which represents $\spe(q_j)$.
   Finally, $\mrp(i + |q_1 \cdots q_{k-1}|)$ has to be the node which represents $\spe(q_k)$ or an its descendant 
   since $q_k$ is a suffix of $p$.
\end{proof}

\begin{theorem}
  Using our augmented $\PPH(\inputtrie)$,
  one can perform parameterized pattern matching queries
  in $O(m \log (\sigma + \pi) + m\pi + \pocc)$ time.
\end{theorem}

%% file: construction.tex
\section{Construction of parameterized position heaps}

In this section, we show how to construct the augmented $\PPH(\inputtrie)$ 
of a given common suffix trie $\inputtrie$ of size $N$.
For convenience, we will sometimes identify each node $v$ of $\PPH(\inputtrie)$ 
with the string which is represented by $v$.
In Section~\ref{subsec:comp_pCST}, we show how to compute $\pCST$ from a given common suffix trie $\inputtrie$.
In Section~\ref{subsec:comp_pph}, we propose how to construct $\PPH(\inputtrie)$ from $\pCST$.

\subsection{Computing $\pCST$ from $\inputtrie$} \label{subsec:comp_pCST}

Here, we show how to construct $\pCST$ of a given $\inputtrie$ of size $N$.

\begin{lemma} \label{lem:comp_pCST}
  For any common suffix trie $\inputtrie$ of size $N$,
  $\pCST$ can be computed in $O(N \pi)$ time and space.
\end{lemma}

\begin{proof}
  We process every node of $\inputtrie$ in a breadth first manner.
  Let $x_j$ be the p-string which is represented by $j$-the node of $\inputtrie$.
  Suppose that we have processed the first $k$ nodes and have computed $\pCST^i$ ($i \leq k$).
  We assume that the $j$-th node of $\inputtrie$, for any $1 \leq j \leq k$, holds
  the resulting substitutions
  from $x_j$ to $\spe((x_j)^R)^R$ (i.e., $x_j[\alpha]$ is mapped to $\spe((x_j)^R)^R[\alpha]$),
  and also a pointer to the corresponding node of $\pCST^i$ (i.e., pointer to the node representing $\spe((x_j)^R)^R$).
  We consider processing the $(k+1)$-th node of $\inputtrie$.
  Since $x_{k+1}$ is encoded from right to left,
  we can determine a character $\spe((x_{k+1})^R)^R[1]$ in $O(\pi)$ time.
  Then, we can insert a new node that represents $\spe((x_{k+1})^R)^R$ 
  as a parent of the node which represents $\spe((x_{k+1}[2..|x_{k+1}|])^R)^R$
  if there does not exist such a node in $\pCST^i$.
  Therefore, we can compute $\pCST$ in $O(N \pi)$ time and space.
\end{proof}

\subsection{Computing $\PPH(\inputtrie)$ from $\pCST$} \label{subsec:comp_pph}

For efficient construction of our $\PPH(\inputtrie)$,
we use \emph{reversed suffix links} defined as follows.

\begin{definition}[Reversed suffix links] \label{def:reversed_suffix_link}
For any node $v$ of $\PPH(\inputtrie)$ and a character $a \in \Sigma \cup \Pi$,
let
\[
 \rslink(a, v) =
 \begin{cases}
   \spe(av) & \mbox{if $\spe(av)$ is represented by $\PPH(\inputtrie)$},\\
   \mbox{undefined} & \mbox{otherwise}.
 \end{cases}
\]
\end{definition}

\begin{figure}[tb]
  \centerline{
    \raisebox{-25mm}{\includegraphics[scale=0.5]{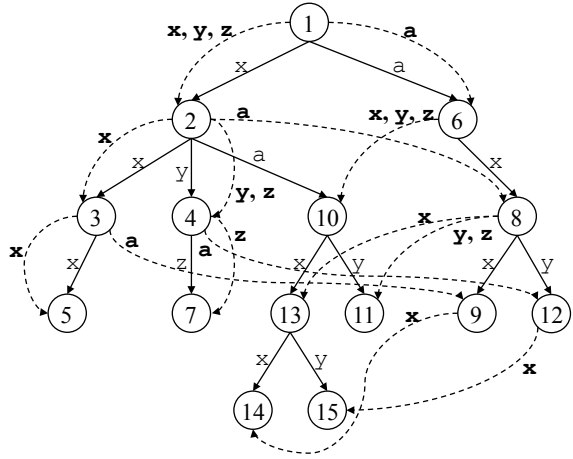}}
  }
  \caption{
    $\PPH(\inputtrie)$ with all reversed suffix links is illustrated in this figure.
    Each dashed arrow shows a reversed suffix link.
    The label of a suffix link is drawn by a bold character.
  }
  \label{fig:rsl-pph}
\end{figure}

See Figure~\ref{fig:rsl-pph} for an example of $\PPH(\inputtrie)$ with reversed suffix links.
In our algorithm, firstly, we insert a new node $h_i$ of $\PPH(\inputtrie)^{i}$ to $\PPH(\inputtrie)^{i-1}$.
After that, we add new suffix links which point to $h_i$.
When we have computed $\PPH(\inputtrie)$, then we compute all maximal reach pointers of $\PPH(\inputtrie)$.

\subsubsection{Inserting a new node}

Assume that $\cnode{j}$ (i.e., $j$-th node of $\pCST$) 
is the child of $\cnode{i}$ for any $2 \leq i \leq N_p$.
Consider to insert $\pnode{i}$ (i.e., the node of $\PPH(\inputtrie)$ which corresponds to $c_i$) 
to $\PPH(\inputtrie)^{i-1}$.
We show how to find the parent of $\pnode{i}$ by starting from $\pnode{j}$.
There are 3 cases based on $w_i[1]$ as follows:
\begin{itemize}
  \item $w_i[1] \in \Pi$ and $w_i[1]$ appears in $w_j[1..|\pnode{j}|]$ (Lemma~\ref{lem:insertion_1}),
  \item $w_i[1] \in \Pi$ and $w_i[1]$ does not appear in $w_j[1..|\pnode{j}|]$ (Lemma~\ref{lem:insertion_2}),
  \item $w_i[1] \in \Sigma$ (Lemma~\ref{lem:insertion_3}).
\end{itemize}

\begin{lemma} \label{lem:insertion_1}
  Assume that $w_i[1] \in \Pi$ appears in $w_j[1..|\pnode{j}|]$, 
  and $a$ is the character in $\Pi$ such that $a = \spe(w_j[1..|\pnode{j}|])[\alpha]$ and $w_i[1] = w_j[\alpha]$ 
  for some $1 \leq \alpha \leq |w_j[1..|\pnode{j}|]|$.
  Let $\pnode{k}$ be the node of $\PPH(\inputtrie)^{i-1}$ 
  which is the lowest ancestor of $\pnode{j}$ 
  that has a reversed suffix link labeled with $a$.
  Then, $\pnode{i}$ is a child of the node representing $\rslink(a, \pnode{k})$.
\end{lemma}

\begin{proof}
  Let $\ell$ be the length of $\rslink(a, \pnode{k})$.
  To prove this lemma, we show that 
  \begin{enumerate}
    \item $\rslink(a, \pnode{k}) = \spe(w_i)[1..\ell]$, and 
    \item There does not exist a node which represents $\spe(w_i)[1..\ell+1]$ in $\PPH(\inputtrie)^{i-1}$.
  \end{enumerate}

  By the definition of reversed suffix links and $\spe$, we have
  \begin{eqnarray*}
  \rslink(a, \pnode{k}) &=& \spe(a \cdot \spe(w_j[1..\ell-1])) = \spe(w_i[1] \cdot w_j[1..\ell-1])\\
                        &=& \spe(w_i[1] \cdot w_i[2..\ell]) = \spe(w_i[1..\ell]).
  \end{eqnarray*}
  Thus, we have proved the first statement.

  By a similar argument,
  we also have $\spe(a \cdot \spe(w_j[1..\ell])) = \spe(w_i[1..\ell+1])$.
  Thus, if $\spe(w_i)[1..\ell+1]$ is represented in $\PPH(\inputtrie)^{i-1}$,
  the node representing $\spe(w_j[1..\ell])$ must have a reversed suffix link labeled with $a$.
  This contradicts the fact that $\pnode{k}$ is the lowest ancestor of $\pnode{j}$ 
  which has a reversed suffix link labeled with $a$.
\end{proof}

\begin{lemma} \label{lem:insertion_2}
  Assume that $w_i[1] \in \Pi$ does not appear in $w_j[1..|\pnode{j}|]$.
  Let $\pnode{k}$ be the node of $\PPH(\inputtrie)^{i-1}$ which is the lowest ancestor of $\pnode{j}$ 
  that has a reversed suffix link labeled with $a \in \Pi \setminus \{\pnode{j}[\alpha] \mid 1 \leq \alpha \leq |\pnode{j}|\}$.
  Then, $\pnode{i}$ is a child of the node representing $\rslink(a, \pnode{k})$.
\end{lemma}

\begin{proof}
  Let $\ell$ be the length of $\rslink(a, \pnode{k})$.
  We show similar statements to the proof of the previous lemma hold,
  but with different assumptions on $w_i[1]$ and $a$.
  By the definition of reversed suffix links and $\spe$, we have
  \[
   \rslink(a, \pnode{k}) = \spe(a \cdot \spe(w_j[1..\ell-1])) = \spe(a \cdot w_j[1..\ell-1]) = \spe(w_i[1..\ell]).
  \]
  Thus, we have proved the first statement.

  By a similar argument,
  we also have $\spe(a \cdot \spe(w_j[1..\ell])) = \spe(w_i[1..\ell+1])$.
  This implies that the second statement holds (similar to the proof of the previous lemma).
\end{proof}

\begin{lemma} \label{lem:insertion_3}
  Assume that $w_i[1] \in \Sigma$.
  Let $\pnode{k}$ be the node in $\PPH(\inputtrie)^{i-1}$ 
  which is the lowest ancestor of $\pnode{j}$ 
  that has a reversed suffix link labeled with $w_i[1]$.
  Then, $\pnode{i}$ is a child of the node representing $\rslink(w_i[1], \pnode{k})$.
\end{lemma}

\begin{proof}  
  Since $w_i[1] \in \Sigma$, we can show the lemma in a similar way to the above proofs.
\end{proof}

\subsubsection{Inserting new reversed suffix links}

In our algorithm, we will add reversed suffix links which point to $h_i$ after inserting a new node $h_i$.
The following lemma shows the number of nodes which point to $h_i$ by reversed suffix links is at most one.

\begin{lemma} \label{lem:rsl-from-one}
  For any node $v$ of $\PPH(\inputtrie)$,
  the number of nodes which point to $v$ by reversed suffix links is at most one.
\end{lemma}

\begin{proof}
  Let $v_1, v_2$ be nodes of $\PPH(\inputtrie)$.
  Assume that $\rslink(a_1, v_1) = \rslink(a_2, v_2)$ for some $a_1, a_2 \in \Sigma \cup \Pi$ 
  and $v_1 \neq v_2$ hold.
  By the definition of reversed suffix links,
  $\spe(a_1 \cdot v_1) = \spe(a_2 \cdot v_2)$.
  Namely, $a_1 \cdot v_1 \approx a_2 \cdot v_2$ holds.
  This implies that $v_1 \approx v_2$, i.e., $\spe(v_1) = \spe(v_2)$.
  Since $v_1$ and $v_2$ are node of $\PPH(\inputtrie)$,
  $\spe(v_1) = v_1$ and $\spe(v_2) = v_2$ hold.
  This contradicts the fact that $v_1 \neq v_2$.
\end{proof}
By the above lemma and arguments of insertion,
the node which points to the new node $\pnode{i}$ by reversed suffix links is 
only a child of $\pnode{k}$ which is an ancestor of $\pnode{j}$.

\subsubsection{Construction algorithm}

Finally, we explain our algorithm of constructing our position heap.
From the above lemmas, we can use similar techniques to Nakashima et al.~\cite{position_heaps_of_trie_2012}
which construct the position heap of a trie of normal strings.
One main difference is the computation of the label of inserted edges/reversed suffix links.
In so doing, each node $h_{\alpha}$ holds the resulting substitutions 
from $w_{\alpha}[1..|h_{\alpha}|]$ to $\spe(w_{\alpha}[1..|h_{\alpha}|])$.
By using these substitutions, we can compute the corresponding label in $O(\pi)$ time.
Thus, we can insert new nodes and new suffix links in $O(\pi)$ time for each node of $\pCST$.
In fact, since we need to use $(\sigma + \pi)$-copies of the position heap for nearest marked ancestor queries on each character, we use $O(\sigma + \pi)$ time to update the data structures needed for each node of $\pCST$.
Therefore, we have the following lemma.

\begin{lemma} \label{lem:comp_pph}
  We can compute $\PPH(\inputtrie)$ from $\pCST$ of size $N_p$
  in $O(N_p(\sigma + \pi))$ time and space.
\end{lemma}

Therefore, we can obtain the following result by Lemmas~\ref{lem:comp_pCST} and~\ref{lem:comp_pph}.

\begin{theorem}
  We can compute $\PPH(\inputtrie)$ of a given common suffix trie $\inputtrie$ of size $N$ 
  in $O(N(\sigma + \pi))$ time and space.
\end{theorem}

Since we can also compute all maximal reach pointers of $\PPH(\inputtrie)$ efficiently 
in a similar way to~\cite{position_heaps_of_trie_2012}
(this algorithm is also similar to suffix link construction),
we also have the following lemma.

\begin{lemma}
  We can compute all the maximal reach pointers for $\PPH(\inputtrie)$ 
  in $O(N_p(\sigma + \pi))$ time and space.
\end{lemma}

Hence, we can get the following result.

\begin{theorem}
  We can compute the augmented $\PPH(\inputtrie)$ of a given common suffix trie $\inputtrie$ of size $N$ 
  in $O(N(\sigma + \pi))$ time and space.
\end{theorem}

%% file: conclusions.tex
\section{Conclusions and open problems}

This paper proposed the p-position heap for a CS trie $\inputtrie$,
denoted $\PPH(\inputtrie)$,
which is the first indexing structure for the p-matching problem on a trie.
The key idea is to transform the input CS trie $\inputtrie$ into
a parameterized CS trie $\pCST$ where p-matching suffixes are merged.
We showed that the p-matching problem on the CS trie $\inputtrie$
can be reduced to the p-matching problem on the parameterized CS trie $\pCST$.
We proposed an algorithm which constructs $\PPH(\inputtrie)$
in $O(N (\sigma + \pi))$ time and working space,
where $N$ is the size of the CS trie $\inputtrie$.
We also showed that using $\PPH(\mathcal{P})$
one can solve the p-matching problem on the CS trie $\inputtrie$
in $O(m \log (\sigma + \pi) + m \pi + \pocc)$ time,
where $m$ is the length of a query pattern and $\pocc$ is the number of occurrences to report.

Examples of open problems regarding this work are the following:
\begin{itemize}
\item Would it be possible to shave
  the $m\pi$ term in the pattern matching time using p-position heaps?
  This $m\pi$ term is introduced when the depth of the corresponding path of 
  $\PPH(\inputtrie)$ is shorter the pattern length $m$
  and thus the pattern needs to be partitioned into $O(\pi)$ blocks
  in the current pattern matching algorithm~\cite{DiptaramaKONS17}.

\item Can we efficiently build the p-suffix tree for a CS trie?
  It is noted by Baker~\cite{Baker93,Baker96} that the \emph{destination} of
  a parameterized suffix link (p-suffix link) of the p-suffix tree can be
  an \emph{implicit node}
  that lies on an edge, and hence there is no monotonicity in the chain of p-suffix links.
  If we follow the approach by Breslauer~\cite{breslauer_suffix_tree_tree_1998}
  which is based on Weiner's algorithm~\cite{Weiner},
  then we need to use the reversed p-suffix link.
  It is, however, unclear whether one can adopt this approach
  since the \emph{origin} of a reversed p-suffix link may be an implicit node.
  Recall that in each step of construction
  we need to find the nearest (implicit) ancestor
  that has a reversed p-suffix link labeled with a given character.
  Since there can be $\Theta(N^2)$ implicit nodes,
  we cannot afford to explicitly maintain information about the reversed p-suffix links
  for all implicit nodes.
\end{itemize}

%% file: main.bbl
\begin{thebibliography}{10}

\bibitem{Baker93}
B.~S. Baker.
\newblock A theory of parameterized pattern matching: algorithms and
  applications.
\newblock In {\em STOC 1993}, pages 71--80, 1993.

\bibitem{Baker96}
B.~S. Baker.
\newblock Parameterized pattern matching: Algorithms and applications.
\newblock {\em J. Comput. Syst. Sci.}, 52(1):28--42, 1996.

\bibitem{breslauer_suffix_tree_tree_1998}
D.~Breslauer.
\newblock The suffix tree of a tree and minimizing sequential transducers.
\newblock {\em Theoretical Computer Science}, 191(1--2):131--144, 1998.

\bibitem{coffman}
E.~Coffman and J.~Eve.
\newblock File structures using hashing functions.
\newblock {\em Communications of the ACM}, 13:427--432, 1970.

\bibitem{DiptaramaKONS17}
Diptarama, T.~Katsura, Y.~Otomo, K.~Narisawa, and A.~Shinohara.
\newblock Position heaps for parameterized strings.
\newblock In {\em Proc. CPM 2017}, pages 8:1--8:13, 2017.

\bibitem{ehrenfeucht_position_heaps_2011}
A.~Ehrenfeucht, R.~M. McConnell, N.~Osheim, and S.-W. Woo.
\newblock Position heaps: A simple and dynamic text indexing data structure.
\newblock {\em Journal of Discrete Algorithms}, 9(1):100--121, 2011.

\bibitem{FujisatoNIBT18}
N.~Fujisato, Y.~Nakashima, S.~Inenaga, H.~Bannai, and M.~Takeda.
\newblock Right-to-left online construction of parameterized position heaps.
\newblock In {\em Proc. PSC 2018}, pages 91--102, 2018.

\bibitem{Kosaraju95}
S.~R. Kosaraju.
\newblock Faster algorithms for the construction of parameterized suffix trees
  (preliminary version).
\newblock In {\em FOCS 1995}, pages 631--637, 1995.

\bibitem{Kucherov13}
G.~Kucherov.
\newblock On-line construction of position heaps.
\newblock {\em J. Discrete Algorithms}, 20:3--11, 2013.

\bibitem{McC76}
E.~M. McCreight.
\newblock A space-economical suffix tree construction algorithm.
\newblock {\em Journal of ACM}, 23(2):262--272, 1976.

\bibitem{MendivelsoP15}
J.~Mendivelso and Y.~Pinz\'{o}n.
\newblock Parameterized matching: Solutions and extensions.
\newblock In {\em Proc. PSC 2015}, pages 118--131, 2015.

\bibitem{position_heaps_of_trie_2012}
Y.~Nakashima, T.~I, S.~Inenaga, H.~Bannai, and M.~Takeda.
\newblock The position heap of a trie.
\newblock In {\em Proc. SPIRE 2012}, volume 7608 of {\em Lecture Notes in
  Computer Science}, pages 360--371, 2012.

\bibitem{NakashimaIIBT15}
Y.~Nakashima, T.~I, S.~Inenaga, H.~Bannai, and M.~Takeda.
\newblock Constructing {LZ78} tries and position heaps in linear time for large
  alphabets.
\newblock {\em Inf. Process. Lett.}, 115(9):655--659, 2015.

\bibitem{Shibuya_construct_stree_of_tree}
T.~Shibuya.
\newblock Constructing the suffix tree of a tree with a large alphabet.
\newblock {\em IEICE Transactions on Fundamentals of Electronics},
  E86-A(5):1061--1066, 2003.

\bibitem{Shibuya04}
T.~Shibuya.
\newblock Generalization of a suffix tree for {RNA} structural pattern
  matching.
\newblock {\em Algorithmica}, 39(1):1--19, 2004.

\bibitem{Ukk95}
E.~Ukkonen.
\newblock On-line construction of suffix trees.
\newblock {\em Algorithmica}, 14(3):249--260, 1995.

\bibitem{Weiner}
P.~Weiner.
\newblock Linear pattern-matching algorithms.
\newblock In {\em Proc. of 14th IEEE Ann. Symp. on Switching and Automata
  Theory}, pages 1--11, 1973.

\end{thebibliography}
